\documentclass[11pt,a4paper]{article}
\pdfoutput=1

\usepackage[utf8]{inputenc}
\usepackage{microtype}
\usepackage[bookmarks]{hyperref}
\hypersetup{colorlinks=true,citecolor=blue,linkcolor=blue,filecolor=blue,urlcolor=blue}
\usepackage{amsmath,amssymb,amsthm}
\usepackage{dsfont}
\usepackage{fullpage}
\usepackage{cleveref}

\usepackage{autonum}
\usepackage{enumerate}
\usepackage{mathtools} 

\allowdisplaybreaks[4]
\newtheorem{theorem}{Theorem} 
\newtheorem{lemma}[theorem]{Lemma}
\newtheorem{proposition}[theorem]{Proposition}

\theoremstyle{definition}

\newtheorem{remark}[theorem]{Remark}

\newcommand{\cB}{\mathcal{B}}

\newcommand{\cD}{\mathcal{D}}

\newcommand{\cH}{\mathcal{H}}

\newcommand{\cK}{\mathcal{K}}
\newcommand{\cL}{\mathcal{L}}
\newcommand{\cM}{\mathcal{M}}
\newcommand{\cN}{\mathcal{N}}

\newcommand{\cP}{\mathcal{P}}

\newcommand{\cR}{\mathcal{R}}
\newcommand{\cS}{\mathcal{S}}

\newcommand{\cX}{\mathcal{X}}

\newcommand{\tD}{\widetilde{D}}
\newcommand{\tH}{\bar{H}}
\newcommand{\tS}{\widetilde{S}}

\newcommand{\tQ}{\widetilde{Q}}
\newcommand{\tsigma}{\tilde{\sigma}}

\newcommand{\one}{\mathds{1}}

\DeclareMathOperator{\tr}{Tr}
\DeclareMathOperator{\rk}{rk}
\DeclareMathOperator{\id}{id}

\DeclareMathOperator{\supp}{supp}
\DeclareMathOperator*{\argmax}{arg\,max}
\DeclareMathOperator*{\argmin}{arg\,min}

\DeclareMathOperator{\spec}{spec}

\newcommand{\sumi}{\sum\nolimits}

\newcommand{\ox}{\otimes}

\newcommand{\doublemid}{\,\middle\|\,}

\usepackage[backend=bibtex8,sorting=nty,sortcites=false,firstinits=true,doi=false,isbn=false,url=false,maxbibnames=10]{biblatex}
\renewbibmacro{in:}{}
\title{Data processing for the sandwiched Rényi divergence:\\ a condition for equality}
\author{Felix Leditzky\thanks{Corresponding author; e-mail: \texttt{f.leditzky@statslab.cam.ac.uk}} \and Cambyse Rouzé \and Nilanjana Datta}
\date{
\textit{\small Statistical Laboratory, Centre for Mathematical Sciences, University of Cambridge,}\\
\textit{\small Wilberforce Road, Cambridge CB3 0WB, United Kingdom} \\[0.4cm]
\today
}

\begin{document}
\maketitle

\begin{abstract}
The $\alpha$-sandwiched Rényi divergence satisfies the data processing inequality, i.e.~monotonicity under quantum operations, for $\alpha\geq 1/2$.
In this article, we derive a necessary and sufficient algebraic condition for equality in the data processing inequality for the $\alpha$-sandwiched Rényi divergence for all $\alpha\geq 1/2$.
For the range $\alpha\in [1/2,1)$, our result provides the only condition for equality obtained thus far.
To prove our result, we first consider the special case of partial trace, and derive a condition for equality based on the original proof of the data processing inequality by Frank and Lieb \cite{FL13} using a strict convexity/concavity argument.
We then generalize to arbitrary quantum operations via the Stinespring Representation Theorem.
As applications of our condition for equality in the data processing inequality, we deduce conditions for equality in various entropic inequalities.
We formulate a Rényi version of the Araki-Lieb inequality and analyze the case of equality, generalizing a result by Carlen and Lieb \cite{CL12} about equality in the original Araki-Lieb inequality.
Furthermore, we prove a general lower bound on a Rényi version of the entanglement of formation, and observe that it is attained by states saturating the Rényi version of the Araki-Lieb inequality.
Finally, we prove that the known upper bound on the entanglement fidelity in terms of the usual fidelity is saturated only by pure states.
\end{abstract}

% Mathematics subject classification: 
% -) 47N50: Operator Theory (47) - Applications in the physical sciences
% -) 81P40  Quantum theory (81) - Quantum coherence, entanglement, quantum correlations
% -) 81P45: Quantum theory (81) - Quantum information, communication, networks
% -) 94A17: Information and communication, circuits (94) - Measures of information, entropy
% -) 94A40: Information and communication, circuits (94) - Channel models (including quantum)


\section{Introduction}\label{sec:intro}
The concept of a relative entropy is fundamental in quantum information theory.
One of the most important examples is the \emph{quantum relative entropy} (QRE), defined for a quantum state $\rho$ and a positive semidefinite operator $\sigma$ with $\supp\rho\subseteq\supp\sigma$ as\footnote{See \Cref{sec:main-result} for definitions and notation.}
\begin{align}\label{eq:quantum-relative-entropy}
D(\rho\|\sigma) \coloneqq \tr(\rho(\log\rho-\log\sigma)).
\end{align}
The QRE satisfies $D(\rho\|\sigma)\geq 0$ if both $\rho$ and $\sigma$ are quantum states, and has an important operational interpretation as a measure of distinguishability of two quantum states, as it characterizes the minimal type-II error in asymmetric quantum hypothesis testing \cite{HP91,ON00}.
In addition, the QRE acts as a parent quantity for various entropic quantities (such as the von Neumann entropy, the conditional entropy, and the Holevo quantity), which characterize the optimal rates of information-theoretic tasks in the asymptotic memoryless setting. 
A key ingredient in establishing these characterizations of information-theoretic tasks is the data processing inequality (DPI), which states that the QRE cannot increase under the joint action of a quantum operation $\Lambda$, 
\begin{align}\label{eq:data-processing-inequality}
D(\rho\|\sigma) \geq D(\Lambda(\rho)\|\Lambda(\sigma)).
\end{align}
Using the operational interpretation of $D(\cdot\|\cdot)$ as a measure of distinguishability, we can interpret \eqref{eq:data-processing-inequality} in the following way: 
A physical transformation (modeled by the quantum operation $\Lambda$) of a quantum system cannot enhance our ability to distinguish between two quantum states $\rho$ and $\sigma$ describing the system.

A natural question, however, is to ask when a quantum operation does not affect the distinguishability of $\rho$ and $\sigma$.
More precisely, given a quantum operation $\Lambda$, we are interested in characterizing those $\rho$ and $\sigma$ for which we have \emph{equality} in the DPI, that is,
\begin{align}\label{eq:dpi-equality}
D(\rho\|\sigma) = D(\Lambda(\rho)\|\Lambda(\sigma)).
\end{align}
The answer to this question was given by Petz \cite{Pet86c,Pet88}, who proved that \eqref{eq:dpi-equality} holds if and only if there exists a \emph{recovery map} given by a quantum operation $\cR$ which reverses the action of $\Lambda$ on $\rho$ and $\sigma$, that is, $\cR(\Lambda(\rho)) = \rho$ and $\cR(\Lambda(\sigma)) = \sigma$ (see \Cref{sec:conclusion} for a precise statement).
This property of $\Lambda$ is also called \emph{sufficiency} \cite{Pet86c,Pet88,MP04,Mos05,JP06,Jen12,HM16,Jen16}.
Petz's result about equality in the DPI has found important applications in quantum information theory.
For example, in \cite{HJPW04} it was used to characterize the case of equality in the strong subadditivity of the von Neumann entropy \cite{LR73}, giving rise to the concept of a short quantum Markov chain.
Moreover, sparked by a breakthrough result by Fawzi and Renner \cite{FR14} relating the notion of recoverability to states with small conditional mutual information, there has been a recent surge of interest in the topic of recoverability \cite{BHOS15,BT15,SFR15,STH15,Wil15,DW15,JRS+15}.
Note that strong subadditivity is equivalent to non-negativity of the quantum conditional mutual information, and hence there is an intimate connection between recoverability and saturation of strong subadditivity.

In general, we call a real-valued functional $\cD(\cdot\|\cdot)$ on pairs of positive semidefinite operators a (generalized) relative entropy if it is non-negative on quantum states and satisfies the DPI $\cD(\rho\|\sigma)\geq \cD(\Lambda(\rho)\|\Lambda(\sigma))$ for any quantum operation $\Lambda$.
An important family of relative entropies is given by the \emph{quantum Rényi divergences}, two important variants of which are known as the \emph{$\alpha$-relative Rényi entropy} ($\alpha$-RRE) and the \emph{$\alpha$-sandwiched Rényi divergence} ($\alpha$-SRD).
These can be seen as special cases of a two-parameter family of relative entropies known as $\alpha$-$z$-Rényi relative entropies \cite{AD15}.
For $\alpha\in (0,\infty)\setminus \lbrace 1\rbrace$ and positive semidefinite operators $\rho$ and $\sigma$, the $\alpha$-RRE $D_\alpha(\rho\|\sigma)$ \cite{Pet86} is defined as
\begin{align}
D_\alpha(\rho\|\sigma) \coloneqq \begin{cases}
\frac{1}{\alpha-1}\log \left\lbrace(\tr\rho)^{-1} \tr\left(\rho^\alpha \sigma^{1-\alpha}\right) \right\rbrace & \!\begin{aligned}[t] &\text{if $\supp\rho\subseteq\supp\sigma$ or} \\ & \text{($\alpha\in (0,1)$ and $\rho \not\perp\sigma$)} \end{aligned} \\[1ex]
+\infty & \text{otherwise}.
\end{cases}
\end{align}
The values at $0$, $1$ and $\infty$ are determined by taking the respective limits, with $\lim_{\alpha\to 1}D_\alpha(\rho\|\sigma) = D(\rho\|\sigma)$.
The $\alpha$-RRE is non-negative for quantum states $\rho$ and $\sigma$, and satisfies the DPI for $\alpha\in[0,2]$ \cite{Lie73,Pet86,Uhl77}.
It has direct operational interpretations as generalized cut-off rates in quantum hypothesis testing \cite{MH11} and error exponents in composite hypothesis testing \cite{HT14}.
Hiai \emph{et al.} \cite{HMPB11} derived necessary and sufficient conditions for equality in the DPI for $D_\alpha(\cdot\|\cdot)$ (and more generally for the class of $f$-divergences).
We discuss this result in \Cref{sec:conclusion}.

The $\alpha$-SRD \cite{MDSFT13,WWY14} is defined for $\alpha\in (0,\infty)\setminus \lbrace 1\rbrace$ and positive semidefinite operators $\rho$ and $\sigma$ as
\begin{align}
\tD_\alpha(\rho\|\sigma) &\coloneqq \begin{cases}
\frac{1}{\alpha-1}\log \left\lbrace(\tr\rho)^{-1} \tr\left[\left(\sigma^{(1-\alpha)/2\alpha} \rho \sigma^{(1-\alpha)/2\alpha} \right)^\alpha \right] \right\rbrace & \!\begin{aligned}[t] &\text{if $\supp\rho\subseteq\supp\sigma$ or} \\ & \text{($\alpha\in (0,1)$ and $\rho \not\perp\sigma$)} \end{aligned} \\[1ex]
+\infty & \text{otherwise}.
\end{cases} 
\end{align}
As before, we define $\tD_*(\cdot\|\cdot)$ for $*\in\lbrace 0,1,\infty\rbrace$ by taking the respective limits, and note that we again have $\lim_{\alpha\to 1}\tD_\alpha(\rho\|\sigma) = D(\rho\|\sigma)$.
However, in general $\tD_0(\rho\|\sigma) \neq D_0(\rho\|\sigma)$ \cite{DL14a}.
Furthermore, $\tD_\infty(\rho\|\sigma)$ coincides with the max-relative entropy $D_{\text{max}}(\rho\|\sigma)$ \cite{Dat09}.
The $\alpha$-SRD satisfies $\tD_\alpha(\rho\|\sigma)\geq 0$ for states $\rho$ and $\sigma$, and has operational interpretations as the strong converse exponent in various settings in quantum hypothesis testing \cite{MO13,CMW14,HT14} and classical-quantum channel coding \cite{MO14}. 
As proved in \cite{FL13,Bei13} (see also \cite{MDSFT13,WWY14,Hia13}), it satisfies the DPI for the range $\alpha\in [1/2,\infty)$: 
\begin{align}\label{eq:srd-dpi}
\tD_\alpha(\rho\|\sigma) \geq \tD_\alpha(\Lambda(\rho)\|\Lambda(\sigma)).
\end{align}
Moreover, there are counterexamples to \eqref{eq:srd-dpi} for the range $\alpha\in (0,1/2)$ \cite{BFT15}.

\section{Main result}\label{sec:main-result}
Our main result in this paper, \Cref{thm:eq-condition} below, is a necessary and sufficient condition for equality in the DPI \eqref{eq:srd-dpi}.
In order to state it properly, we first introduce some necessary notation and terminology.

Throughout this paper we only consider finite-dimensional Hilbert spaces.
All logarithms are taken to base 2.
For a Hilbert space $\cH$ we write $\cB(\cH)$ for the algebra of linear operators on $\cH$, and we denote by $\cP(\cH)\coloneqq \lbrace \rho\in\cB(\cH)\colon \rho\geq 0\rbrace$ and $\cD(\cH)\coloneqq \lbrace \rho\in\cP(\cH)\colon \tr\rho=1\rbrace$ the sets of positive semidefinite operators and density matrices (or quantum states), respectively.
We denote by $\rk A$ the rank of an operator $A$, and by $\supp A$ the support of $A$, i.e.~the orthogonal complement of the kernel of $A$.
We write $A \not\perp B$ if $\supp A \cap \supp B$ contains at least one non-zero vector.
For a Hermitian operator $A$, we denote by $\spec A\subseteq \mathbb{R}$ the set of eigenvalues of $A$.
For a pure state $|\psi\rangle\in\cH$ we write $\psi=|\psi\rangle\langle\psi|\in\cD(\cH)$ for the corresponding rank-$1$ density matrix.
Given a linear map $\cL\colon \cB(\cH)\to\cB(\cK)$ between Hilbert spaces $\cH$ and $\cK$, the adjoint map $\cL^\dagger \colon\cB(\cK)\to \cB(\cH)$ is the unique map satisfying $\langle \cL^\dagger(Y),X\rangle = \langle Y,\cL(X)\rangle$ for all $X\in\cB(\cH)$ and $Y\in\cB(\cK)$, where $\langle A,B\rangle\coloneqq \tr(A^\dagger B)$ is the Hilbert-Schmidt inner product.
A linear map $\Phi\colon \cB(\cH)\to\cB(\cK)$ between Hilbert spaces $\cH$ and $\cK$ is called $n$-positive, if $\id_n\ox\Phi\colon \cB(\mathbb{C}^n)\ox\cB(\cH) \to \cB(\mathbb{C}^n)\ox\cB(\cK)$ is positive, where $\id_n$ denotes the identity map on $\cB(\mathbb{C}^n)$.
A map is completely positive if it is $n$-positive for all $n\in\mathbb{N}$.
A quantum operation (or quantum channel) $\Lambda$ between Hilbert spaces $\cH$ and $\cK$ is a linear, completely positive, and trace-preserving map $\Lambda\colon \cB(\cH)\to\cB(\cK)$.

Our main result is given by the following theorem:
\begin{theorem}\label{thm:eq-condition}
Let $\alpha\in [1/2,1)\cup(1,\infty)$ and set $\gamma=(1-\alpha)/2\alpha$.
Furthermore, let $\rho\in\cD(\cH)$ and $\sigma\in\cP(\cH)$ with $\supp\rho\subseteq\supp\sigma$ if $\alpha>1$ or $\rho\not\perp\sigma$ if $\alpha < 1$, and let $\Lambda\colon \cB(\cH) \to \cB(\cK)$ be a quantum operation.
We have equality in the data processing inequality \eqref{eq:srd-dpi}, 
\begin{align}
\tD_\alpha(\rho\|\sigma) = \tD_\alpha(\Lambda(\rho)\|\Lambda(\sigma)),
\end{align}
if and only if
\begin{align}
\sigma^\gamma \left(\sigma^\gamma \rho \sigma^\gamma \right)^{\alpha-1} \sigma^\gamma = \Lambda^\dagger \!\left( \Lambda(\sigma)^\gamma \left[\Lambda(\sigma)^\gamma \Lambda(\rho) \Lambda(\sigma)^\gamma \right]^{\alpha-1} \Lambda(\sigma)^\gamma \right).
\end{align}
\end{theorem}
For $\alpha > 1$ and \emph{positive} trace-preserving maps, \Cref{thm:eq-condition} was also proved using the framework of non-commutative $L_p$-spaces \cite{DJW15}.
The case of equality in the DPI for the $\alpha$-SRD was also discussed in two papers by Hiai and Mosonyi \cite{HM16} and Jenčová \cite{Jen16}, both of which focused on the aspect of sufficiency.
The connections between \Cref{thm:eq-condition} and these results are discussed in \Cref{sec:conclusion}.
The rest of this paper is organized as follows:
In \Cref{sec:proof-of-dpi}, we analyze the proof of the DPI \eqref{eq:srd-dpi} for the $\alpha$-SRD as given in \cite{FL13}, extracting a necessary and sufficient condition for equality in \eqref{eq:srd-dpi} and thus proving \Cref{thm:eq-condition}.
We present applications of \Cref{thm:eq-condition} to entanglement and distance measures in \Cref{sec:applications}.
Finally, in \Cref{sec:conclusion} we compare our result to the recoverability/sufficiency results mentioned above, and state some open questions.

\section{Proof of the main result}\label{sec:proof-of-dpi}
For the remainder of the discussion we will assume that $\rho\in\cD(\cH)$ and $\sigma\in\cP(\cH)$ with $\supp\rho\subseteq\supp\sigma$ if $\alpha>1$, or $\rho\not\perp\sigma$ if $\alpha\in[1/2,1)$.
We set $\gamma=(1-\alpha)/2\alpha$ and define the trace functional
\begin{align}
\tQ_\alpha(\rho\|\sigma)\coloneqq \tr\!\left[\left(\sigma^\gamma\rho \sigma^\gamma\right)^\alpha\right]\!,
\end{align}
which is invariant under joint unitary conjugation and tensoring with an arbitrary state as follows: For any unitary $U$ and any state $\tau$, we have
\begin{align}
\tQ_\alpha\!\left(U\rho U^\dagger\|U\sigma U^\dagger\right) &= \tQ_\alpha(\rho\|\sigma), \label{eq:TF-invariance-unitary}\\
\tQ_\alpha(\rho\ox\tau \|\sigma\ox\tau) &= \tQ_\alpha(\rho \| \sigma) .\label{eq:TF-invariance-tensor}
\end{align}
The $\alpha$-SRD can be expressed in terms of this trace functional as $\tD_\alpha(\rho\|\sigma) = \frac{1}{\alpha-1}\log \tQ_\alpha(\rho\|\sigma)$, and hence, $\tD_\alpha(\cdot\|\cdot)$ inherits the invariance properties \eqref{eq:TF-invariance-unitary} and \eqref{eq:TF-invariance-tensor} from $\tQ_\alpha(\cdot\|\cdot)$.
By virtue of the Stinespring Representation Theorem \cite{Sti55}, the DPI \eqref{eq:srd-dpi} is thus equivalent to monotonicity of the $\alpha$-SRD under partial trace:
\begin{align}\label{eq:srd-partial-trace-mon}
\tD_\alpha(\rho_{AB}\|\sigma_{AB}) \geq \tD_\alpha(\rho_A\|\sigma_A),
\end{align}
where the subscripts $AB$ and $A$ indicate the Hilbert spaces $\cH_{AB}=\cH_A\ox\cH_B$ and $\cH_A$ on which the density matrices act, and the partial trace is taken over the $B$ system.
Since the logarithm is monotonically increasing, the monotonicity of $\tD_\alpha(\cdot\|\cdot)$ under partial trace \eqref{eq:srd-partial-trace-mon} is in turn equivalent to the following monotonicity properties of $\tQ_\alpha(\cdot\|\cdot)$:
\begin{align}
\begin{aligned}
\tQ_\alpha(\rho_{AB}\|\sigma_{AB}) &\leq \tQ_\alpha(\rho_A\|\sigma_A)& &\text{for $\alpha\in[1/2,1)$},\\
\tQ_\alpha(\rho_{AB}\|\sigma_{AB}) &\geq \tQ_\alpha(\rho_A\|\sigma_A)& &\text{for $\alpha\in(1,\infty)$}.
\end{aligned}\label{eq:TF-partial-trace}
\end{align}
We set $d=\dim\cH_B$ and let $\lbrace V_i \rbrace_{i=1}^{d^2}$ be a representation of the discrete Heisenberg-Weyl group on $\cH_B$, satisfying the following relation (see e.g.~\cite{Wil13} or \cite{Wol12}):
\begin{align}
\frac{1}{d^2}\sum_i (\one_A\ox V_i) \rho_{AB} (\one_A\ox V_i^\dagger) &= \rho_A\ox \pi_B\label{eq:HW-average}
\end{align}
where $\pi_B = \one_B/d$ denotes the completely mixed state on $B$.
The crucial ingredient in proving \eqref{eq:TF-partial-trace} is then the joint concavity/convexity of the trace functional $\tQ_\alpha(\cdot\|\cdot)$:

\begin{proposition}[{\cite{FL13}}]\label{prop:TF-joint-convexity}
The functional $(\rho,\sigma)\mapsto \tQ_\alpha(\rho\|\sigma)$ is jointly concave for $\alpha\in [1/2,1)$ and jointly convex for $\alpha\in(1,\infty)$.
\end{proposition}
\begin{remark}
The joint convexity/concavity of the trace functional $\tQ_\alpha(\cdot\|\cdot)$ is a special case of the joint convexity/concavity of a more general trace functional underlying the $\alpha$-z-Rényi relative entropies mentioned in \Cref{sec:intro}, which was proved by \textcite{Hia13} using the theory of Pick functions.
A more accessible proof can be found in the arXiv version of \cite{AD15}.
\end{remark}

Joint convexity/concavity of the trace functional $\tQ_\alpha(\cdot\|\cdot)$ as stated in \Cref{prop:TF-joint-convexity} can be used to prove the monotonicity of $\tQ_\alpha(\cdot\|\cdot)$ under partial trace \eqref{eq:TF-partial-trace} as follows. 
Abbreviating $V_i\equiv\one_A\ox V_i$, we have for $\alpha>1$ that
\begin{align}
\tQ_\alpha(\rho_{AB}\|\sigma_{AB}) &= \tQ_\alpha\! \left(V_i\rho_{AB} V_i^\dagger \doublemid V_i \sigma_{AB} V_i^\dagger\right) \label{eq:proof-of-TF-partial-trace-init}\\
&= \frac{1}{d^2} \sumi_i \tQ_\alpha\!\left(V_i\rho_{AB} V_i^\dagger \doublemid V_i \sigma_{AB} V_i^\dagger\right)\\
&\geq \tQ_\alpha\left(d^{-2} \sumi_i V_i\rho_{AB} V_i^\dagger \doublemid d^{-2} \sumi_i V_i \sigma_{AB} V_i^\dagger \right)\label{eq:dpi-proof-ineq}\\
&= \tQ_\alpha ( \rho_A\ox\pi_B\| \sigma_A\ox \pi_B)\\
&= \tQ_\alpha (\rho_A \| \sigma_A).\label{eq:proof-of-TF-partial-trace-final}
\end{align}
In the first equality we used the invariance of $\tQ_\alpha(\cdot\|\cdot)$ under joint unitary conjugation \eqref{eq:TF-invariance-unitary}. The inequality follows from the joint convexity of $\tQ_\alpha(\cdot\|\cdot)$ as stated in \Cref{prop:TF-joint-convexity}. In the third equality we used property \eqref{eq:HW-average} of the Heisenberg-Weyl operators, and in the last equality we used the invariance of $\tQ_\alpha(\cdot\|\cdot)$ under tensoring with a fixed state \eqref{eq:TF-invariance-tensor}.
For $\alpha\in[1/2,1)$, we go through the same steps as above to show that
\begin{align}
\tQ_\alpha(\rho_{AB}\|\sigma_{AB}) \leq \tQ_\alpha(\rho_A\|\sigma_A),
\end{align}
only this time employing the joint concavity of $\tQ_\alpha(\cdot\|\cdot)$ from \Cref{prop:TF-joint-convexity} in \eqref{eq:dpi-proof-ineq}.

To derive an equality condition for \eqref{eq:TF-partial-trace} (and hence, \eqref{eq:srd-dpi}), we take a closer look at \Cref{prop:TF-joint-convexity}.
The key ingredient in its proof in \cite{FL13} is to rewrite the trace functional $\tQ_\alpha(\rho\|\sigma)$ as follows:
Defining the function 
\begin{align}\label{eq:f-alpha}
f_\alpha(H,\rho,\sigma)\coloneqq \alpha \tr\rho H - (\alpha-1) \tr\left[\left( \sigma^{-\gamma} H \sigma^{-\gamma}\right)^{\alpha/(\alpha-1)}\right]\!,
\end{align}
it holds that
\begin{align}\label{eq:Q-as-sup}
\tQ_\alpha(\rho\|\sigma) = \begin{cases}
\inf\nolimits_{H\geq 0} f_\alpha(H,\rho,\sigma) & \text{if $\alpha\in[1/2,1)$}\\
\sup\nolimits_{H\geq 0} f_\alpha(H,\rho,\sigma) & \text{if $\alpha>1$.}\\
\end{cases}
\end{align}
The joint concavity/convexity of $\tQ_\alpha(\cdot\|\cdot)$ then follows from showing that $(\rho,\sigma)\mapsto f_\alpha(H,\rho,\sigma)$ is jointly concave/convex for fixed $H$ and the respective ranges of $\alpha$.
Moreover, in the course of proving the validity of \eqref{eq:Q-as-sup}, \textcite{FL13} show that for fixed $\rho,\sigma$ a critical point of $f_\alpha(H,\rho,\sigma)$ satisfying $\partial f_\alpha(H,\rho,\sigma)/\partial H = 0$ is given by
\begin{align}\label{eq:optimal-H}
\hat{H} = \sigma^\gamma (\sigma^\gamma \rho \sigma^\gamma )^{\alpha-1} \sigma^\gamma.
\end{align}
As $H\mapsto f_\alpha(H,\rho,\sigma)$ is concave for $\alpha>1$ and convex for $\alpha\in[1/2,1)$, the critical point $\hat{H}$ in \eqref{eq:optimal-H} is a maximum of $f_\alpha(H,\rho,\sigma)$ for $\alpha>1$ and fixed $\rho$ and $\sigma$, and a minimum of $f_\alpha(H,\rho,\sigma)$ for $\alpha\in[1/2,1)$ and fixed $\rho$ and $\sigma$.
Consequently, it holds that
\begin{align}\label{eq:TF-with-optimizer}
\tQ_\alpha(\rho,\sigma) = f_\alpha(\hat{H},\rho,\sigma)\qquad \text{for all $\alpha \in [1/2,1)\cup (1,\infty)$.}
\end{align}
In the following, we show that $H\mapsto f_\alpha(H,\rho)$ is in fact \emph{strictly} concave/convex, such that the optimizer $\hat{H}$ in \eqref{eq:optimal-H} is a \emph{unique} maximizer/minimizer.
To this end, we employ the following result, which is proved e.g.~in \cite[Thm.~2.10]{Car10}.
\begin{theorem}\label{thm:trace-functions-convexity}
Let $A$ be a Hermitian matrix with $\spec A\subseteq \cD\subseteq \mathbb{R}$, and let $g\colon \cD\to\mathbb{R}$ be a continuous, (strictly) convex function. 
Then the function $A\mapsto \tr g(A)$ is (strictly) convex.
\end{theorem}
Let us first consider $\alpha>1$.
The function $H\mapsto \tr[( \sigma^{-\gamma} H \sigma^{-\gamma})^{\alpha/(\alpha-1)}]$ is the composition of the linear function $X\mapsto \sigma^{-\gamma}X\sigma^{-\gamma}$ and the functional $A\mapsto \tr A^{\alpha/(\alpha-1)}$, the latter being strictly convex by \Cref{thm:trace-functions-convexity} upon choosing $g\colon \mathbb{R}_+\to \mathbb{R}_+,\, g(x)= x^{\alpha/(\alpha-1)}$.
As $\alpha>1$, the function $H\mapsto -(\alpha-1)\tr[( \sigma^{-\gamma} H \sigma^{-\gamma})^{\alpha/(\alpha-1)}]$ is therefore strictly concave, and hence, $f_\alpha(H,\rho,\sigma)$ is strictly concave, since it is the sum of a linear function and a strictly concave function.
In the case $\alpha\in[1/2,1)$, a similar argument shows that $f_\alpha(H,\rho,\sigma)$ is strictly convex.

We have seen in \eqref{eq:dpi-proof-ineq} above that the joint concavity/convexity of the trace functional $\tQ_\alpha(\cdot\|\cdot)$ is the only step in the proof of \eqref{eq:srd-partial-trace-mon} involving an inequality.
Let us analyze this step further.
For $\alpha>1$, we abbreviate $\rho_i=V_i\rho_{AB}V_i^\dagger$, $\sigma_i=V_i\sigma_{AB} V_i^\dagger$, and $\lambda_i=d^{-2}$ for $i=1,\dots,d^2$, such that $\rho = \sum_i\lambda_i \rho_i = \rho_A\ox\pi_B$ and $\sigma = \sum_i\lambda_i\sigma_i = \sigma_A\ox\pi_B$ by \eqref{eq:HW-average}.
We then consider the operators
\begin{align}
\tH &\coloneqq \argmax\nolimits_H f_\alpha(H,\rho,\sigma) &
H_i &\coloneqq \argmax\nolimits_H f_\alpha(H, \rho_i,\sigma_i),
\end{align}
which are well-defined by the preceding discussion.
Step \eqref{eq:dpi-proof-ineq} above can now be written as
\begin{align}\label{eq:joint-convexity-with-f}
\tQ_\alpha(\rho\|\sigma) = f_\alpha\!\left(\tH,\rho,\sigma\right) \leq \sumi_i \lambda_i f_\alpha\!\left(\tH, \rho_i,\sigma_i\right) \leq \sumi_i \lambda_i f_\alpha\!\left(H_i,\rho_i,\sigma_i\right) = \sumi_i \lambda_i \tQ_\alpha(\rho_i\|\sigma_i).
\end{align}
Assume now that we have equality in the joint convexity, that is, $\tQ_\alpha(\rho\|\sigma) = \sumi_i \lambda_i \tQ_\alpha(\rho_i\|\sigma_i)$.
Then the chain of inequalities in \eqref{eq:joint-convexity-with-f} collapses, and in particular we obtain 
\begin{align}
f_\alpha\!\left(\tH, \rho_i,\sigma_i\right) = f_\alpha\!\left(H_i,\rho_i,\sigma_i\right)\quad \text{for every $i=1,\dots,d^2$}.
\end{align}
In other words, the operator $\tH$ maximizes $f_\alpha(H,\rho_i,\sigma_i)$ for every $i=1,\dots,d^2$, and since the maximizing element of $f_\alpha(H,\rho_i,\sigma_i)$ is unique, we obtain $\tH = H_i$ for every $i=1,\dots,d^2$.
In the case $\alpha\in [1/2,1)$, we define $\tH \coloneqq \argmin\nolimits_H f_\alpha(H,\rho,\sigma)$ and $H_i \coloneqq \argmin\nolimits_H f_\alpha(H, \rho_i,\sigma_i)$. 
The inequalities in \eqref{eq:joint-convexity-with-f} are now reversed due to the joint concavity of $\tQ_\alpha(\cdot\|\cdot)$, and since $f_\alpha(\cdot,\rho_i,\sigma_i)$ attains a \emph{minimum} at $H_i$. 
Again, we obtain $\tH = H_i$ for every $i=1,\dots,d^2$.

Using the explicit form of the optimal $\hat{H}$ from \eqref{eq:optimal-H}, we can write out the condition $\tH= H_i$ with the choices for $\rho_i$, $\sigma_i$, and $\lambda_i$ made above, obtaining for every $i=1,\dots,d^2$
\begin{align}\label{eq:Hbar-Hi}
\sigma_A^\gamma \left(\sigma_A^\gamma \rho_A \sigma_A^\gamma\right)^{\alpha-1} \sigma_A^\gamma \ox \pi_B^{2\gamma + (\alpha-1)(2\gamma+1)} = (\one_A\ox V_i) \sigma_{AB}^\gamma \left(\sigma_{AB}^\gamma \rho_{AB} \sigma_{AB}^\gamma\right)^{\alpha-1} \sigma_{AB}^\gamma \left(\one_A\ox V_i^\dagger\right)\!.
\end{align}
Since $2\gamma + (\alpha-1)(2\gamma + 1) = 0$, the dimension factor of $\pi_B$ cancels, and eliminating the unitary $V_i$ in \eqref{eq:Hbar-Hi} yields
\begin{align}\label{eq:eq-condition}
\sigma_A^\gamma \left(\sigma_A^\gamma \rho_A \sigma_A^\gamma \right)^{\alpha-1} \sigma_A^\gamma \ox \one_B = \sigma_{AB}^\gamma \left(\sigma_{AB}^\gamma \rho_{AB} \sigma_{AB}^\gamma \right)^{\alpha-1} \sigma_{AB}^\gamma.
\end{align}
This is a necessary condition for equality in the monotonicity of the trace functional $\tQ_\alpha(\cdot\|\cdot)$.
Furthermore, it is easy to see that \eqref{eq:eq-condition} is also sufficient, as $\tQ_\alpha(\rho_{AB}\|\sigma_{AB}) = \tQ_\alpha(\rho_A\|\sigma_B)$ follows from multiplying \eqref{eq:eq-condition} by $\rho_{AB}$, taking the trace, and using cyclicity of the trace.
In summary, we have therefore proved the following:

\begin{proposition}\label{prop:eq-condition-partial-trace}
Let $\alpha\in [1/2,1)\cup(1,\infty)$, then we have $\tD_\alpha(\rho_{AB}\|\sigma_{AB}) = \tD_\alpha(\rho_A\|\sigma_A)$ if and only if
\begin{align}
\sigma_A^\gamma \left(\sigma_A^\gamma \rho_A \sigma_A^\gamma \right)^{\alpha-1} \sigma_A^\gamma \ox \one_B = \sigma_{AB}^\gamma \left(\sigma_{AB}^\gamma \rho_{AB} \sigma_{AB}^\gamma \right)^{\alpha-1} \sigma_{AB}^\gamma.
\end{align}
\end{proposition}

We are now in a position to prove our main result, \Cref{thm:eq-condition}:

\begin{proof}[Proof of \Cref{thm:eq-condition}]
For the quantum operation $\Lambda\colon \cB(\cH)\to\cB(\cK)$, the Stinespring Representation Theorem \cite{Sti55} asserts that there is a Hilbert space $\cH'$, a pure state $|\tau\rangle\in\cH'\ox\cK$, and a unitary $U$ acting on $\cH\ox\cH'\ox\cK$ such that for every $\rho\in\cB(\cH)$ we have
\begin{align}
\Lambda(\rho) = \tr_{12}\left(U(\rho\ox\tau)U^\dagger\right)\!,
\end{align}
where $\tr_{12}$ denotes the partial trace over $\cH$ and $\cH'$, that is, the first two factors of $\cH\ox\cH'\ox\cK$.
We then have
\begin{align}
\tD_\alpha(\rho\|\sigma) &= \tD_\alpha\!\left(U(\rho\ox\tau)U^\dagger \doublemid U(\sigma\ox\tau)U^\dagger \right)\\
&\geq \tD_\alpha\!\left(\tr_{12}\left(U(\rho\ox\tau)U^\dagger\right) \doublemid \tr_{12}\left(U(\sigma\ox\tau)U^\dagger\right)\right)\\
&= \tD_\alpha(\Lambda(\rho)\|\Lambda(\sigma)),
\end{align}
where the first line follows from \eqref{eq:TF-invariance-unitary} and \eqref{eq:TF-invariance-tensor}, and the inequality follows from \eqref{eq:srd-partial-trace-mon}.
By \Cref{prop:eq-condition-partial-trace} we have equality in the second line if and only if
\begin{multline}
\left( U(\sigma\ox\tau)U^\dagger\right)^{\gamma} \left[\left( U(\sigma\ox\tau)U^\dagger\right)^\gamma U(\rho\ox\tau)U^\dagger \left( U(\sigma\ox\tau)U^\dagger\right)^\gamma\right]^{\alpha-1} \left( U(\sigma\ox\tau)U^\dagger\right)^{\gamma} \\
 = \one_{\cH\ox\cH'} \ox \Lambda(\sigma)^{\gamma}\left[\Lambda(\sigma)^\gamma \Lambda(\rho)\Lambda(\sigma)^\gamma\right]^{\alpha-1} \Lambda(\sigma)^\gamma.\label{eq:tr-12}
\end{multline}
Using the fact that $f(UXU^\dagger)=Uf(X)U^\dagger$ for every function $f$ and unitary $U$, this is equivalent to
\begin{align}
U\left(\sigma^{\gamma}\left(\sigma^\gamma \rho \sigma^\gamma \right)^{\alpha-1} \sigma^\gamma \ox \tau \right) U^\dagger  =  \one_{\cH\ox\cH'} \ox \Lambda(\sigma)^{\gamma}\left[\Lambda(\sigma)^\gamma \Lambda(\rho)\Lambda(\sigma)^\gamma\right]^{\alpha-1} \Lambda(\sigma)^\gamma.
\end{align}
The theorem now follows from the fact that the adjoint of $\Lambda$ is given by $\Lambda^\dagger(\omega) = V^\dagger(\one_{\cH\ox\cH'} \ox \omega)V$, where $V = U(\one_\cH\ox |\tau\rangle)$ is the Stinespring isometry of $\Lambda$ satisfying $V^\dagger V = \one_{\cH}$.
\end{proof}

\section{Applications}\label{sec:applications}
In this section we discuss applications of \Cref{thm:eq-condition}.
Our goal is to generalize a set of results by Carlen and Lieb \cite{CL12} about the Araki-Lieb inequality and entanglement of formation by proving the corresponding results for Rényi quantities.
In \Cref{sec:renyi-araki-lieb} we state a Rényi version of the Araki-Lieb inequality (\Cref{prop:renyi-araki-lieb}), and analyze the case of equality (\Cref{thm:equality-in-renyi-araki-lieb}).
In \Cref{sec:EoF} we first prove a general lower bound on the Rényi entanglement of formation (analogous to the corresponding bound on the entanglement of formation in \cite{CL12}), and then use the results from \Cref{sec:renyi-araki-lieb} to show that this lower bound is achieved by states saturating the Rényi version of the Araki-Lieb inequality.
These results are presented in \Cref{thm:renyi-EoF}.
Finally, in \Cref{sec:entanglement-fidelity} we discuss the case of equality in a well-known upper bound on the entanglement fidelity in terms of the usual fidelity, which we state in \Cref{prop:entanglement-fidelity-equality}.

We start with a few definitions.
For $\rho\in\cD(\cH)$ the \emph{von Neumann entropy} $S(\rho)$ is defined as $S(\rho)\coloneqq -\tr (\rho\log\rho) = -D(\rho\|\one)$, and we write $S(A)_\rho\equiv S(\rho_A)$ for a state $\rho_A$ acting on a Hilbert space $\cH_A$.
The \emph{conditional entropy} $S(A|B)_\rho$ is defined as $S(A|B)_\rho \coloneqq S(AB)_\rho - S(B)_\rho = -D(\rho_{AB}\|\one_A\ox\rho_B)$.
In our discussion we consider the following Rényi generalization of the conditional entropy, first defined in \cite{MDSFT13}: 
For $\alpha\in [1/2,\infty)$, the \emph{$\alpha$-Rényi conditional entropy} of a bipartite state $\rho_{AB}$ is defined as
\begin{align}
\tS_\alpha(A|B)_\rho \coloneqq -\min_{\sigma_B} \tD_\alpha(\rho_{AB}\|\one_A\ox\sigma_B),
\end{align}
where the minimization is over states $\sigma_B$, and we set $\tS_1(A|B)_\rho\coloneqq \lim_{\alpha\to 1}\tS_\alpha(A|B)_\rho = S(A|B)_\rho$.
The $\alpha$-Rényi conditional entropy satisfies the following duality relation:
\begin{proposition}[{\cite{MDSFT13,Bei13}}]\label{prop:renyi-conditional-duality}
Let $\rho_{ABC}$ be a pure state with marginals $\rho_{AB}$ and $\rho_{AC}$. 
For $\alpha,\beta\in [1/2,\infty)$ such that $\frac{1}{\alpha} + \frac{1}{\beta} = 2$, we have
\begin{align}
\tS_\alpha(A|B)_\rho = -\tS_\beta(A|C)_\rho.
\end{align}
\end{proposition}

\subsection{Rényi version of Araki-Lieb inequality and the case of equality}\label{sec:renyi-araki-lieb}
The Araki-Lieb inequality \cite{AL70} states that for every bipartite state $\rho_{AB}$,
\begin{align}\label{eq:araki-lieb}
S(AB)_\rho \geq \left| S(A)_\rho - S(B)_\rho\right|.
\end{align}
There are a few different characterizations for the case of equality in the Araki-Lieb inequality \cite{NC00,ZW11,CL12}.
Here, we concentrate on a result by Carlen and Lieb:
\begin{theorem}[Equality in the Araki-Lieb inequality; {\cite{CL12}}]\label{thm:equality-in-araki-lieb}
For a bipartite state $\rho_{AB}$ denote by $r_{AB}$, $r_A$, and $r_B$ the ranks of $\rho_{AB}$, $\rho_A$, and $\rho_B$, respectively.
The state $\rho_{AB}$ saturates the Araki-Lieb inequality \eqref{eq:araki-lieb}, 
\begin{align}
S(AB)_\rho = S(B)_\rho - S(A)_\rho,
\end{align} 
if and only if the following conditions are satisfied:
\begin{enumerate}[{\normalfont (i)}]
\item $r_B = r_A r_{AB}$
\item The state $\rho_{AB}$ has a spectral decomposition of the form
\begin{align}
\rho_{AB} = \sum_{i=1}^{r_{AB}} \lambda_i |i\rangle\langle i|_{AB},
\end{align}
where the vectors $\lbrace |i\rangle_{AB}\rbrace_{i=1}^{r_{AB}}$ are such that $\tr_B|i\rangle \langle j|_{AB} = \delta_{ij}\rho_A$ for $i,j=1,\dots,r_{AB}$.
\end{enumerate}
\end{theorem}
We can regard the Araki-Lieb inequality \eqref{eq:araki-lieb} as lower bounds on the conditional entropies:
\begin{align}
S(A|B)_\rho &\geq - S(A)_\rho & S(B|A)_\rho &\geq - S(B)_\rho.\label{eq:lower-bound-on-cond-entropy}
\end{align}
In the following, we only focus on the bound $S(A|B)_\rho \geq - S(A)_\rho$, noting that all the results we obtain hold for $S(B|A)_\rho$ in an analogous manner.
The formulation \eqref{eq:lower-bound-on-cond-entropy} of the Araki-Lieb inequality admits a simple proof based on duality as follows.
With a purification $|\rho\rangle_{ABC}$ of $\rho_{AB}$, we have
\begin{align}
S(A|B)_\rho &= - S(A|C)_\rho = D(\rho_{AC}\|\one_A\ox\rho_C) \geq D(\rho_A\|\one_A) = -S(A)_\rho,\label{eq:araki-lieb-conditional-proof}
\end{align}
where the first equality follows from duality for the conditional entropy, and the inequality follows from the DPI for the QRE with $\Lambda = \tr_C$.

The advantage of phrasing the Araki-Lieb inequality in the form of \eqref{eq:lower-bound-on-cond-entropy} is that we can easily generalize it to Rényi quantities.
To this end, we simply replace the von Neumann quantities in \eqref{eq:araki-lieb-conditional-proof} by the Rényi conditional entropy and the Rényi entropy, defined as 
\begin{align}
S_\alpha(A)_\rho = -\tD_\alpha(\rho_A\|\one_A) = \frac{1}{1-\alpha} \log\tr \rho_A^\alpha.\label{eq:renyi-entropy}
\end{align} 
With $\alpha,\beta\in[1/2,\infty)$ such that $\frac{1}{\alpha} + \frac{1}{\beta} = 2$, we then have
\begin{align}\label{eq:renyi-araki-lieb-conditional-proof}
\tS_\alpha(A|B)_\rho = -\tS_\beta(A|C)_\rho = \tD_\beta(\rho_{AC}\|\one_A\ox\tsigma_C) \geq \tD_\beta(\rho_A\|\one_A) = -S_\beta(A)_\rho.
\end{align}
Here, we used \Cref{prop:renyi-conditional-duality} in the first equality, chose an optimizing state $\tsigma_C$ for $\tS_\beta(A|C)_\rho$ in the second equality, and the inequality is simply the DPI for the $\beta$-SRD with respect to $\Lambda=\tr_C$.
We also obtain the upper bound $\tS_\alpha(A|B)_\rho \leq S_\alpha(A)_\rho$ by a simple application of the DPI with respect to $\Lambda=\tr_B$.
Hence, we have proved:
\begin{lemma}[Rényi version of the Araki-Lieb inequality]\label{prop:renyi-araki-lieb}
Let $\rho_{AB}$ be a bipartite state, and $\alpha,\beta\in[1/2,\infty)$ be such that $\frac{1}{\alpha} + \frac{1}{\beta} = 2$, then
\begin{align}\label{eq:renyi-araki-lieb}
-S_\beta(A)_\rho \leq \tS_\alpha(A|B)_\rho \leq S_\alpha(A)_\rho.
\end{align}
\end{lemma}

Since the inequality in the lower bound of \Cref{prop:renyi-araki-lieb} stems from the DPI for $\tD_\beta(\cdot\|\cdot)$ (cf.~\eqref{eq:renyi-araki-lieb-conditional-proof}), we can apply \Cref{thm:eq-condition} (in the form of \Cref{prop:eq-condition-partial-trace}) to investigate the case of equality.
By \Cref{prop:eq-condition-partial-trace} we have $\tD_\beta(\rho_{AC}\|\one_A\ox\tsigma_C) = \tD_\beta(\rho_A\|\one_A)$ if and only if
\begin{align}\label{eq:eq-in-renyi-araki-lieb}
\rho_A^{\beta-1} \ox \one_C = \left(\one_A\ox\tsigma_C^\delta\right) \left( \left(\one_A\ox\tsigma_C^\delta\right) \rho_{AC} \left(\one_A\ox\tsigma_C^\delta\right) \right)^{\beta-1} \left(\one_A\ox\tsigma_C^\delta\right)\!,
\end{align}
where $\delta = (1-\beta)/2\beta$.
It is easy to see that \eqref{eq:eq-in-renyi-araki-lieb} is equivalent to $\rho_{AC} = \rho_A\ox\tsigma_C$, that is, if $\rho_{ABC}$ is a purification of $\rho_{AB}$, then the marginal $\rho_{AC}$ is of product form.
We can then go through the same steps as in the proof of Theorem 1.4 in \cite{CL12} (which we stated as \Cref{thm:equality-in-araki-lieb} above) to arrive at the following Rényi generalization of this result:
\begin{theorem}[Equality in the Rényi version of the Araki-Lieb inequality]\label{thm:equality-in-renyi-araki-lieb}
Let $\rho_{AB}$ be a bipartite state with purification $\rho_{ABC}$, and let $\alpha,\beta\in[1/2,\infty)$ be such that $1/\alpha + 1/\beta = 2$. 
Denote by $r_{AB}$, $r_{A}$, and $r_B$ the ranks of $\rho_{AB}$, $\rho_A$, and $\rho_B$, respectively.
We have equality in the Rényi version of the Araki-Lieb inequality, 
\begin{align}
\tS_\alpha(A|B)_\rho = -S_\beta(A)_\rho,
\end{align} 
if and only if the following conditions are satisfied:
\begin{enumerate}[{\normalfont (i)}]
\item $r_B = r_A r_{AB}.$

\item The state $\rho_{AB}$ has a spectral decomposition of the form
\begin{align}
\rho_{AB} = \sum_{i=1}^{r_{AB}} \lambda_i |i\rangle\langle i|_{AB},
\end{align}
where the vectors $\lbrace |i\rangle_{AB}\rbrace_{i=1}^{r_{AB}}$ are such that $\tr_B|i\rangle \langle j|_{AB} = \delta_{ij}\rho_A$ for $i,j=1,\dots,r_{AB}$.
\end{enumerate}
\end{theorem}

\begin{remark}
For the upper bound $\tS_\alpha(A|B)_\rho \leq S_\alpha(A)_\rho$ in \Cref{prop:renyi-araki-lieb}, we have equality if and only if $\tD_\alpha(\rho_{AB}\|\one_A\ox\tsigma_B) = \tD_\alpha(\rho_A\|\one_A)$, where $\tsigma_B$ is a state optimizing $\tS_\alpha(A|B)_\rho$.
Similar to above, we obtain from \Cref{prop:eq-condition-partial-trace} that this is the case if and only if $\rho_{AB} = \rho_A\ox \tsigma_B$.
\end{remark}

\subsection{Rényi entanglement of formation}\label{sec:EoF}
Let $\rho_{AB}$ be a bipartite state, then the \emph{entanglement of formation} (EoF) $E_F(\rho_{AB})$ \cite{BDSW96,BBP+96} is defined as the least expected entropy of entanglement of any ensemble of pure states realizing $\rho_{AB}$:
\begin{align}\label{eq:EoF}
E_F(\rho_{AB}) \coloneqq \min_{\lbrace p_i,\psi_i\rbrace} \sumi_i p_i S(\tr_B{\psi_i}),
\end{align}
where the minimum is over ensembles of pure states $\lbrace p_i,\psi_i\rbrace$ such that $\rho_{AB} = \sum_i p_i |\psi_i\rangle \langle \psi_i|$.
This entanglement measure satisfies $E_F(\rho_{AB})\geq 0$ for all $\rho_{AB}$, and is furthermore \emph{faithful}, that is, $E_F(\rho_{AB}) = 0$ if and only if $\rho_{AB}$ is separable.
The EoF is an upper bound on the (two-way) distillable entanglement \cite{BDSW96,BBP+96}.
Moreover, its regularized version $E_F^\infty(\rho_{AB})\coloneqq \lim_{n\to \infty}E_F(\rho_{AB}^{\ox n})/n$ is equal to the asymptotic entanglement cost of preparing the state $\rho_{AB}$ \cite{HHT01}.
Carlen and Lieb \cite{CL12} prove the following result, which provides a lower bound on $E_F(\rho_{AB})$ that is achieved by states saturating the Araki-Lieb inequality (\Cref{thm:equality-in-araki-lieb}):
\begin{theorem}[{\cite{CL12}}]\label{thm:EoF}
Let $\rho_{AB}$ be a bipartite state.
Then
\begin{align}\label{eq:EoF-lower-bound}
E_F(\rho_{AB}) \geq \max\left\lbrace -S(A|B)_\rho, -S(B|A)_\rho,0\right\rbrace\!,
\end{align}
and this bound is saturated by states satisfying the conditions of \Cref{thm:equality-in-araki-lieb}.
That is, for states $\rho_{AB}$ with $S(A|B)_\rho = -S(A)_\rho$, we have
\begin{align}
E_F(\rho_{AB}) = -S(A|B)_\rho.
\end{align}
\end{theorem}
\begin{remark}
If $S(A|B)_\rho = -S(A)_\rho$, then $E_F(\rho_{AB}) = -S(A|B)_\rho \geq -S(B|A)_\rho$ by \eqref{eq:EoF-lower-bound}.
\end{remark}

Using the results of \Cref{sec:renyi-araki-lieb}, our goal in this section is to obtain a Rényi generalization of \Cref{thm:EoF}.
To this end, we consider the \emph{Rényi entanglement of formation} (REoF) $E_{F,\alpha}(\rho_{AB})$ \cite{Vid00,WMVF16}, which is obtained from the definition of $E_F(\rho_{AB})$ in \eqref{eq:EoF} by replacing the von Neumann entropy with the Rényi entropy of order $\alpha\geq 0$ (as defined in \eqref{eq:renyi-entropy}):
\begin{align}
E_{F,\alpha}(\rho_{AB}) \coloneqq \min_{\lbrace p_i,\psi_i\rbrace} \sumi_i p_i S_\alpha(\tr_B{\psi_i})
\end{align}
Note that in \cite{SBW14} the authors consider a different Rényi generalization of the EoF based on the $\alpha$-Rényi conditional entropy.
As in the von Neumann case, the REoF satisfies $E_{F,\alpha}(\rho_{AB})\geq 0$ for all $\rho_{AB}$, and it is faithful as well.
We prove the following generalization of \Cref{thm:EoF} for $\alpha>1$:
\begin{theorem}\label{thm:renyi-EoF}
Let $\rho_{AB}$ be a bipartite state, and let $\alpha>1$ and $\beta = \alpha/(2\alpha-1)\in(1/2,1)$ such that $1/\alpha + 1/\beta = 2$.
Then we have the following bound on the REoF:
\begin{align}\label{eq:renyi-EoF-lower-bound}
E_{F,\alpha}(\rho_{AB}) \geq \max\left\lbrace -\tS_\beta(A|B)_\rho, -\tS_\beta(B|A)_\rho,0\right\rbrace\!.
\end{align}
If $\rho_{AB}$ saturates the Rényi version \eqref{eq:renyi-araki-lieb} of the Araki-Lieb inequality with Rényi parameter $\beta$, that is, $\tS_\beta(A|B)_\rho = -S_\alpha(A)_\rho$, then
\begin{align}\label{eq:renyi-EoF-saturation}
E_{F,\alpha}(\rho_{AB}) = -\tS_\beta(A|B)_\rho.
\end{align}
\end{theorem}

\begin{proof}
Let $\lbrace q_i, \phi_i\rbrace$ be an ensemble of pure states minimizing the REoF, that is, 
$
E_{F,\alpha}(\rho_{AB}) = \sum_i q_i S_\alpha(\tr_B(\phi_i)).
$
We define a purification $\rho_{ABC}$ of $\rho_{AB}$ by $|\rho\rangle_{ABC} = \sum_i \sqrt{q_i} |\phi_i\rangle_{AB}|i\rangle_C$, where $\lbrace |i\rangle_C\rbrace_i$ is an orthonormal basis for $\cH_C$.
Denoting by $\tsigma_C$ the state optimizing the Rényi conditional entropy $\tS_\alpha(A|C)_\rho$, we have
\begin{align}
\tS_\beta(A|B)_\rho &= -\tS_\alpha(A|C)_\rho\\
&= \tD_\alpha(\rho_{AC}\|\one_A\ox\tsigma_C)\\
&= \tD_\alpha\!\left(\sumi_{i,j} \sqrt{q_i q_j} \tr_B|\phi_i\rangle\langle\phi_j|_{AB}\ox |i\rangle\langle j|_{C} \doublemid \one_A\ox\tsigma_C\right),
\end{align}
where the first line follows from \Cref{prop:renyi-conditional-duality}.
We now apply the pinching map $\rho\mapsto \sum_i |i\rangle \langle i|_C\rho |i\rangle\langle i|_C$ (which is a quantum operation) to both states and use the DPI for $\tD_\alpha(\cdot\|\cdot)$. 
Setting $\lambda_i = \langle i|\tsigma_C|i\rangle$, we obtain
\begin{align}
\tS_\beta(A|B)_\rho &\geq \tD_\alpha\!\left(\sumi_i q_i \tr_B\phi_i \ox |i\rangle\langle i|_C \doublemid \one_A \ox \sumi_i \lambda_i |i\rangle \langle i|_C\right)\\
&= \frac{1}{\alpha-1} \log \left\lbrace\tr\left[\sumi_i q_i^\alpha \lambda_i^{1-\alpha} \left(\tr_B\phi_i\right)^\alpha\ox |i\rangle\langle i|_C\right]\right\rbrace\\
&= \frac{1}{\alpha-1} \log \left\lbrace\sumi_i q_i \left(q_i/\lambda_i\right)^{\alpha-1} \tr\left[\left(\tr_B\phi_i\right)^{\alpha}\right]\right\rbrace\\
&\geq \frac{1}{\alpha-1} \sumi_i q_i \log\left\lbrace \left(q_i/\lambda_i\right)^{\alpha-1} \tr\left[\left(\tr_B\phi_i\right)^{\alpha}\right]\right\rbrace\\
&= \sumi_i q_i \frac{1}{\alpha-1}\log\tr\left[\left(\tr_B\phi_i\right)^{\alpha}\right] + \sumi_i q_i \log (q_i/\lambda_i)\\
&= - E_{F,\alpha}(\rho_{AB}) + D\!\left(\lbrace q_i\rbrace\|\lbrace \lambda_i\rbrace\right)\\
&\geq - E_{F,\alpha}(\rho_{AB}).
\end{align}
In the first equality we used the fact that the states $\sumi_i q_i \tr_B\phi_i \ox |i\rangle\langle i|_C$ and $\one_A \ox \sumi_i \lambda_i |i\rangle \langle i|_C$ commute, and hence $\tD_\alpha(\cdot\|\cdot)$ reduces to the ordinary $\alpha$-RRE $D_\alpha(\omega\|\tau) = \frac{1}{\alpha-1}\log\tr(\omega^\alpha \tau^{1-\alpha})$.
In the second inequality we used concavity of the logarithm together with $\alpha>1$, and in the last inequality we used non-negativity of the classical Kullback-Leibler divergence, defined for probability distributions $P$ and $Q$ on an alphabet $\cX$ as $D(P\|Q) = \sum_{x\in \cX} P(x)\log P(x)/Q(x)$, provided that $P(x)=0$ whenever $Q(x)=0$.
Note that the latter is satisfied as $\supp\rho_C\subseteq\supp\tsigma_C$ holds for the optimizing state $\tsigma_C$ of $\tS_\alpha(A|C)_\rho$ \cite{MDSFT13}.
The bound $E_{F,\alpha}(\rho_{AB})\geq -\tS_\beta(B|A)_\rho$ follows in an analogous way, yielding \eqref{eq:renyi-EoF-lower-bound}.

To prove \eqref{eq:renyi-EoF-saturation}, we first note that by \Cref{thm:equality-in-renyi-araki-lieb} the state $\rho_{AB}$ satisfies $\tS_\beta(A|B)_\rho = -S_\alpha(A)_\rho$ if and only if the rank condition of \Cref{thm:equality-in-renyi-araki-lieb}(i) holds and $\rho_{AB}$ has a spectral decomposition of the form
\begin{align}
\rho_{AB} = \sumi_i \lambda_i |i\rangle\langle i|_{AB},
\end{align}
where the vectors $\lbrace |i\rangle_{AB} \rbrace$ satisfy $\tr_B|i\rangle\langle j|_{AB} = \delta_{ij}\rho_A$.
We can now employ the same argument used in \cite{CL12} in the proof of the second assertion of \Cref{thm:EoF}, to prove the corresponding assertion of \Cref{thm:renyi-EoF}:
\begin{align}
E_{F,\alpha}(\rho_{AB}) &= \min_{\lbrace p_i,\psi_i\rbrace} \sumi_i p_i S_\alpha(\tr_B\psi_i)\\
&\leq \sumi_i \lambda_i S_\alpha(\tr_B|i\rangle\langle i|_{AB})\\
&= \sumi_i \lambda_i S_\alpha(A)_\rho\\
&= S_\alpha(A)_\rho\\
&= -\tS_\beta(A|B)_\rho
\end{align}
where in the inequality we chose the particular ensemble $\lbrace \lambda_i, |i\rangle_{AB}\rbrace$ realizing $\rho_{AB}$.
This upper bound on $E_{F,\alpha}(\rho_{AB})$, together with the general lower bound in \eqref{eq:renyi-EoF-lower-bound}, yields the claim.
\end{proof}

\begin{remark}~
\begin{enumerate}[(i)]
\item The proof method of the lower bound \eqref{eq:renyi-EoF-lower-bound} for $E_{F,\alpha}(\rho_{AB})$ in \Cref{thm:renyi-EoF} can be specialized to the quantum relative entropy $D(\cdot\|\cdot)$, providing a new proof of \eqref{eq:EoF-lower-bound} in \Cref{thm:EoF}:
\begin{align}
S(A|B)_\rho &= - S(A|C)_\rho\\
&= D(\rho_{AC}\|\one_A\ox\rho_C)\\
&= D\!\left(\sumi_{i,j} \sqrt{q_i q_j} \tr_B|\phi_i\rangle\langle\phi_j|_{AB}\ox |i\rangle\langle j|_{C} \doublemid \one_A\ox\rho_C\right)\\
&\geq D\!\left(\sumi_i q_i \tr_B\phi_i \ox |i\rangle\langle i|_C \doublemid \one_A \ox \sumi_i \lambda_i |i\rangle \langle i|_C\right)\\
&= D(\lbrace q_i\rbrace\|\lbrace \lambda_i\rbrace) + \sumi_i q_i D\!\left(\tr_B\phi_i\|\one_A\right)\\
&= D(\lbrace q_i\rbrace\|\lbrace \lambda_i\rbrace) - \sumi_i q_i S(\tr_B\phi_i)\\
&\geq - E_F(\rho_{AB}).
\end{align}
The bound $S(B|A)_\rho \geq - E_F(\rho_{AB})$ can be proved in an analogous way.

\item If a state $\rho_{AB}$ satisfies $\tS_\beta(A|B)_\rho = -S_\alpha(A)_\rho$ for $a>1$ and $\beta = \alpha/(2\alpha-1)$, then $E_{F,\alpha}(\rho_{AB}) = -\tS_\beta (A|B)_\rho \geq -\tS_\beta(B|A)_\rho$ by \eqref{eq:renyi-EoF-lower-bound} in \Cref{thm:renyi-EoF}.
\end{enumerate}
\end{remark}

\subsection{Entanglement fidelity}\label{sec:entanglement-fidelity}
For a state $\rho\in\cD(\cH)$ and a quantum channel $\cN\colon \cB(\cH)\to\cB(\cH)$, the \emph{entanglement fidelity} $F_e(\rho,\cN)$ \cite{Sch96} is defined as
\begin{align}\label{eq:entanglement-fidelity}
F_e(\rho,\cN) \coloneqq \langle\psi^\rho|(\cN\ox \id_{\cH'})(\psi^\rho)|\psi^\rho\rangle,
\end{align}
where the state $|\psi^\rho\rangle\in\cH\ox\cH'$ purifies $\rho$.
Since any two purifications of $\rho$ are related by an isometry acting on the purifying system, the definition \eqref{eq:entanglement-fidelity} of the entanglement fidelity is independent of the chosen purification.
For a mixed state $\rho$ with spectral decomposition $\rho = \sum_{i=1}^{\rk\rho} \lambda_i |i\rangle\langle i|_{\cH}$, a canonical purification is given by
\begin{align}
|\psi^\rho \rangle = \sum_{i=1}^{\rk\rho} \sqrt{\lambda_i} |i\rangle_\cH \ox |i\rangle_{\cH'}
\end{align}
for suitable orthonormal vectors $\lbrace |i\rangle_{\cH'}\rbrace_{i=1}^{\rk \rho}$ in $\cH'$.
Hence, in the following discussion we can assume without loss of generality that $\dim\cH' = \rk\rho$.

The entanglement fidelity $F_e(\rho,\cN)$ can be expressed in terms of the usual fidelity $F(\omega,\tau)\coloneqq \|\sqrt{\omega}\sqrt{\tau}\|_1$ as\footnote{For an arbitrary operator $A$, the trace norm is defined as $\|A\|_1 = \tr\sqrt{A^\dagger A}$.}
\begin{align}
F_e(\rho,\cN) = F(\psi^\rho,(\cN\ox \id_{\cH'})(\psi^\rho))^2.
\end{align}
We have $\|\sqrt{\omega}\sqrt{\tau}\|_1 = \tr(\sqrt{\tau} \omega \sqrt{\tau} )^{1/2}$ by definition of the trace norm, and hence the fidelity is related to the $1/2$-SRD via
\begin{align}\label{eq:fidelity-renyi}
F(\omega,\tau) = \tQ_{1/2}(\omega\|\tau).
\end{align}
It follows from the DPI for $\tQ_{1/2}(\cdot\|\cdot)$ that the fidelity is non-decreasing under partial trace.\footnote{Note that this can also be proved directly, e.g.~via Uhlmann's Theorem \cite{Uhl76}.}
This can be used to prove the following upper bound on the entanglement fidelity, where we write $\cN(\psi^\rho)\equiv (\cN\ox\id_{\cH'})(\psi^\rho)$:
\begin{align}
F_e(\rho,\cN) &= F(\psi^\rho, \cN(\psi^\rho))^2 \leq F(\rho,\cN(\rho))^2
\label{eq:fidelity-bound}
\end{align}

Due to \eqref{eq:fidelity-bound}, the entanglement fidelity provides a more stringent notion of distance between quantum states than the fidelity.
However, it is clear that we have equality in \eqref{eq:fidelity-bound} if the state $\rho$ is pure.
Using our condition for equality in the DPI for the $1/2$-SRD from \Cref{thm:eq-condition} (resp.~\Cref{prop:eq-condition-partial-trace}), we can prove that this is in fact the only case of equality:
\begin{proposition}\label{prop:entanglement-fidelity-equality}
Let $\rho\in\cD(\cH)$ and $\cN\colon\cB(\cH)\to\cB(\cH)$ be a quantum channel, then we have
\begin{align}
F_e(\rho,\cN) = F(\rho,\cN(\rho))^2
\end{align}
if and only if $\rho$ is pure.
\end{proposition}
\begin{proof}
We have already noted above that purity of $\rho$ is sufficient for equality in \eqref{eq:fidelity-bound}.
If $F_e(\rho,\cN) = F(\rho,\cN(\rho))^2$, then \eqref{eq:fidelity-renyi} implies that we have equality in the DPI for the $1/2$-SRD with respect to $\Lambda=\tr_{\cH'}$.
Hence, \Cref{prop:eq-condition-partial-trace} yields
\begin{multline}
\cN(\rho)^{1/2} \left(\cN(\rho)^{1/2}\rho \cN(\rho)^{1/2}\right)^{-1/2} \cN(\rho)^{1/2} \ox \one_{\cH'}\\ = \cN(\psi^\rho)^{1/2} \left(\cN(\psi^\rho)^{1/2} \psi^\rho \cN(\psi^\rho)^{1/2}\right)^{-1/2} \cN(\psi^\rho)^{1/2},
\end{multline}
from which we obtain
\begin{align}
\rho \ox \one_{\cH'} = c(\rho,\cN) \cN(\rho)^{-1} \cN(\psi^\rho) \psi^\rho \cN(\psi^\rho) \cN(\rho)^{-1}\label{eq:ent-fid-eq-condition}
\end{align}
for a suitable constant $c(\rho,\cN)$.
Note that the right-hand side of \eqref{eq:ent-fid-eq-condition} has rank $1$, and hence $\rho\ox\one_{\cH'}$ is a pure state.
But this is only possible if $\rho$ is pure and $\dim\cH'=1$.
\end{proof}

\section{Conclusion and open questions}\label{sec:conclusion}
We have shown that equality in the DPI for the $\alpha$-SRD $\tD_\alpha(\cdot\|\cdot)$ holds for a quantum operation $\Lambda$, a quantum state $\rho$, and a positive semidefinite operator $\sigma$ (with suitable support conditions) if and only if the following algebraic condition is satisfied (setting $\gamma=(1-\alpha)/2\alpha$):
\begin{align}\label{eq:main-result}
\sigma^\gamma \left(\sigma^\gamma \rho \sigma^\gamma \right)^{\alpha-1} \sigma^\gamma = \Lambda^\dagger \!\left( \Lambda(\sigma)^\gamma \left[\Lambda(\sigma)^\gamma \Lambda(\rho) \Lambda(\sigma)^\gamma \right]^{\alpha-1} \Lambda(\sigma)^\gamma \right)
\end{align}
In the case of the $\alpha$-RRE $D_\alpha(\cdot\|\cdot)$ for $\alpha\in[0,2]$ (which includes the QRE corresponding to $\alpha=1$), a necessary and sufficient algebraic condition for equality in the DPI is given by \cite{Pet86c,Pet88,HMPB11}
\begin{align}\label{eq:qre-algebraic-condition}
\Lambda^\dagger\!\left(\Lambda(\sigma)^{-z}\Lambda(\rho)^z\right) = \sigma^{-z}\rho^{z}\quad\text{for all $z\in\mathbb{C}$.}
\end{align}
This can be rephrased in terms of the existence of a recovery map by an argument detailed in \cite{HMPB11}:
We have $D_\alpha(\rho\|\sigma) = D_\alpha(\Lambda(\rho)\|\Lambda(\sigma))$
if and only if there is a recovery map in the form of a quantum operation $\cR_{\sigma,\Lambda}$ such that
\begin{align}
\cR_{\sigma,\Lambda}(\Lambda(\sigma)) &= \sigma & \cR_{\sigma,\Lambda}(\Lambda(\rho)) &= \rho.\label{eq:recovery}
\end{align}
In general, a quantum operation $\Lambda$ is called \emph{sufficient} with respect to a set $\cS\subseteq\cD(\cH)$ of quantum states, if there exists a quantum operation $\cR$ satisfying $\cR(\Lambda(\tau)) = \tau$ for all $\tau\in\cS$.
Hence, \eqref{eq:recovery} says that $\Lambda$ is sufficient for $\lbrace \rho,\sigma\rbrace$.
Furthermore, the particular map $\cR_{\sigma, \Lambda}$ admits an explicit formula on the support of $\Lambda(\sigma)$:
\begin{align}\label{eq:recovery-map}
\cR_{\sigma,\Lambda}(\cdot) = \sigma^{1/2} \Lambda^\dagger\!\left(\Lambda(\sigma)^{-1/2} \cdot \Lambda(\sigma)^{-1/2}\right) \sigma^{1/2}.
\end{align}
Since $\cR_{\sigma,\Lambda}(\Lambda(\sigma)) = \sigma$ holds by definition \eqref{eq:recovery-map} of the recovery map, the non-trivial part of \eqref{eq:recovery} is the assertion $\cR_{\sigma,\Lambda}(\Lambda(\rho)) = \rho$.
Note that by a theorem of Petz \cite{Pet88} a quantum channel $\Lambda$ is sufficient for $\lbrace \rho,\sigma\rbrace$ if and only if $\cR_{\sigma,\Lambda}(\Lambda(\rho)) = \rho$ holds for the map defined in \eqref{eq:recovery-map}.
We also observe that the recovery map $\cR_{\sigma,\Lambda}$ is \emph{independent of $\alpha$}, and the existence of a map $\cR$ satisfying \eqref{eq:recovery} forces equality in the DPI for any $\alpha\in[0,2]$,
\begin{align}\label{eq:double-dpi}
D_\alpha(\rho\|\sigma) \geq D_\alpha(\Lambda(\rho)\|\Lambda(\sigma)) \geq D_\alpha(\cR(\Lambda(\rho))\|\cR(\Lambda(\sigma))) = D_\alpha(\rho\|\sigma),
\end{align}
where the first inequality follows from applying the DPI with respect to $\Lambda$, and the second follows from applying the DPI with respect to $\cR$.
Thus, we have equality in the DPI for the $\alpha$-RRE for \emph{all} $\alpha\in[0,2]$ once it holds for \emph{some} $\alpha\in[0,2]$.

Taking a closer look at the condition \eqref{eq:main-result} for equality in the DPI for the $\alpha$-SRD, it is easy to see that choosing $\alpha=2$ in \eqref{eq:main-result}  yields precisely the statement $\cR_{\sigma,\Lambda}(\Lambda(\rho)) = \rho$.
Hence, in the case $\alpha=2$ we have equality in the DPI for the $2$-SRD if and only if the recovery map $\cR_{\sigma,\Lambda}$ defined in \eqref{eq:recovery-map} satisfies \eqref{eq:recovery}.
This was already observed in \cite{DJW15} for positive trace-preserving maps.

Shortly after completion of the present paper, the connection between sufficiency and equality in the DPI for the $\alpha$-SRD was presented by Jenčová \cite{Jen16} and Hiai and Mosonyi \cite{HM16}.
The main result of \cite{Jen16} is that a $2$-positive trace-preserving map $\Lambda$ is sufficient with respect to $\lbrace\rho,\sigma\rbrace$ if and only if $\tD_\alpha(\Lambda(\rho)\|\Lambda(\sigma)) = \tD_\alpha(\rho\|\sigma)$ holds for some $\alpha >1$.
By the theorem of Petz \cite{Pet88} mentioned above, we therefore have equality in the DPI for the $\alpha$-SRD for any $\alpha>1$ if and only if the map $\cR_{\sigma,\Lambda}$ defined in \eqref{eq:recovery-map} satisfies \eqref{eq:recovery}.
Furthermore, a similar argument as in \eqref{eq:double-dpi} for $\tD_\alpha(\cdot\|\cdot)$ shows that equality holds in the DPI for the $\alpha$-SRD for all $\alpha>1$ if it holds for some $\alpha>1$.
This result settles the sufficiency question for the $\alpha$-SRD for the range $\alpha>1$ and $2$-positive trace-preserving maps (which include all quantum operations).
In \cite{HM16} sufficiency is analyzed for $2$-positive \emph{bistochastic} maps $\Lambda$, that is, both $\Lambda$ and $\Lambda^\dagger$ are $2$-positive and trace-preserving.
The main theorem of \cite{HM16} regarding the $\alpha$-SRD states conditions for sufficiency of $\Lambda$ for certain ranges of $\alpha$ (including the range $\alpha\in[1/2,1)$) under the additional assumption that one of the two states $\rho$ and $\sigma$ is a \emph{fixed point} of $\Lambda$.
In fact, this result is obtained as a corollary of a more general theorem analyzing sufficiency for the $\alpha$-$z$-Rényi relative entropies under similar assumptions.

In view of our main result (\Cref{thm:eq-condition}) and the results of \cite{HM16} and \cite{Jen16}, it remains an open question whether equality in the DPI for the $\alpha$-SRD in the range $\alpha\in(1/2,1)$ is equivalent to sufficiency of $\Lambda$ in our setting, in which $\Lambda$ is an arbitrary quantum operation and $\rho$ and $\sigma$ are states with $\rho \not\perp\sigma$.
Note that for $\alpha=1/2$ it is known that such a general sufficiency result cannot hold \cite{MO13}.\footnote{This can be seen as follows:
The $1/2$-SRD can be expressed in terms of the fidelity as $\tD_{1/2}(\rho\|\sigma) = -2\log F(\rho,\sigma)$.
It is well-known that for given $\rho$ and $\sigma$ there exists a measurement $M=\lbrace M_x\rbrace_{x\in\cX}$ for some finite set $\cX$ such that the fidelity $F(\rho,\sigma)$ is equal to the \emph{classical} fidelity of the measurement outcomes $\lbrace\tr(M_x\rho)\rbrace_{x\in\cX}$ and $\lbrace\tr(M_x\sigma)\rbrace_{x\in\cX}$ obtained from measuring $\rho$ and $\sigma$ (see e.g.~\cite{Wil13}).
Hence, for any two states $\rho$ and $\sigma$ we have equality in the DPI for the $1/2$-SRD with respect to the quantum operation $\cM(\omega) = \sum_{x\in\cX} \tr(\omega M_x) |x\rangle\langle x|$, and it is impossible to recover the states $\rho$ and $\sigma$ from the measurement outcomes alone.
This proves that a general sufficiency result as stated above cannot hold for $\alpha=1/2$.}

Regarding our results in \Cref{sec:applications} about entanglement measures and distances, it would be interesting to see whether the entropic bounds proved therein can be used to characterize information-theoretic tasks.

\paragraph{Acknowledgements.}
ND is grateful to Anna Jenčová and Mark M.~Wilde for earlier discussions on the issue of equality in the DPI for the $\alpha$-SRD in the $L_p$-space framework during the workshop `Beyond IID in Information Theory' (5-10 July 2016) in Banff, Canada.
The authors would also like to thank Will Matthews and Michał Studziński for interesting discussions.

\printbibliography[title={References},heading=bibintoc]
\end{document}